\documentclass[copyright,creativecommons]{eptcs}
\providecommand{\event}{Infinity 2011} 
\usepackage{breakurl}             

\title{Trees over Infinite Structures \\ and Path Logics with Synchronization}
\author{Alex Spelten \qquad Wolfgang Thomas \qquad Sarah Winter
\institute{RWTH Aachen University\\ Germany}
\email{\{spelten,thomas,winter\}@automata.rwth-aachen.de}
}
\def\titlerunning{Trees over Infinite Structures}
\def\authorrunning{A. Spelten, W. Thomas \& S. Winter}

\usepackage{gastex}
\usepackage{amsmath,amssymb}
\usepackage{stmaryrd}
\usepackage{float}
\usepackage{subfigure}
\usepackage{graphicx}
\usepackage{multirow} 
\usepackage[amsmath,thmmarks]{ntheorem}
\newtheorem{theorem}{Theorem}

\newtheorem{lemma}[theorem]{Lemma}
\newtheorem{proposition}[theorem]{Proposition}

\theoremstyle{plain}
\theorembodyfont{\normalfont}

\newtheorem{remark}[theorem]{Remark}
\theoremstyle{nonumberplain}
\theoremsymbol{{\footnotesize\ensuremath{\blacksquare}}}
\newtheorem{proof}{Proof}

\newcommand{\eat}[1]{}

\renewcommand{\emptyset}{\protect\raisebox{.15ex}{\ensuremath{\not}}\ensuremath{\mathnormal{O}}}

\DeclareMathSymbol{\epsilon}{\mathalpha}{letters}{"22}
\DeclareMathSymbol{\varepsilon}{\mathalpha}{letters}{"0F}
\DeclareMathSymbol{\theta}{\mathalpha}{letters}{"23}
\DeclareMathSymbol{\vartheta}{\mathalpha}{letters}{"12}
\DeclareMathSymbol{\phi}{\mathalpha}{letters}{"27}
\DeclareMathSymbol{\varphi}{\mathalpha}{letters}{"1E}
\makeatletter
\newcommand{\charfusion}[2]{%
  \def\ch@rfusion##1##2{%
    \ooalign{\hfil$##1#1$\hfil\cr\hfil$##2$\hfil\crcr}}%
      \mathop{%
      \vphantom{#1}%
      \mathpalette\ch@rfusion#2}\displaylimits}
\makeatother

\newcommand{\m}{\ensuremath{\mathcal{M}}}
\newcommand{\el}{\ensuremath{\mathcal{L}}}
\newcommand{\ml}{\ensuremath{\m\textrm{-}\el}}

\newcommand{\Sing}{\ensuremath{\textrm{Sing}}}
\newcommand{\Succ}{\ensuremath{\textrm{Succ}}}
\newcommand{\Mweak}{\ensuremath{\m^{\#}}}
\newcommand{\MweakE}{\ensuremath{\m^{\#}_E}}
\newcommand{\Mstrong}{\ensuremath{\m^*}}
\newcommand{\MstrongE}{\ensuremath{\m^*_E}}

\newcommand{\tp}{\ensuremath{\mathit{tp}}}
\newcommand{\Suc}{\ensuremath{\mathit{Suc}}}

\def\msos1s{MSO\ensuremath{^{\mathrm{S1S}}}}

\begin{document}
\maketitle

\begin{abstract}
We provide decidability and undecidability results on 
the model-checking problem for infinite tree structures. 
These tree structures are built from sequences of  
elements of infinite relational structures. More 
precisely, we deal with the tree iteration of a relational 
structure $\m$ in the sense of Shelah-Stupp. 
In contrast to classical results, where 
model-checking is shown decidable for MSO-logic, we 
show decidability of the tree model-checking problem 
for logics that allow only path quantifiers and chain 
quantifiers (where chains are subsets of paths), as they appear 
in branching time logics; however, at the same time,  
the tree is enriched by the equal-level relation 
(which holds between vertices $u$, $v$ if they are on 
the same tree level). We separate cleanly the tree logic 
from the logic used for expressing properties of the underlying 
structure $\m$. 
We illustrate the scope of the decidability results by showing 
that two slight extensions of the framework lead to undecidability.
In particular, this applies to the (stronger) tree iteration in 
the sense of Muchnik-Walukiewicz. 

\end{abstract}


\section{Introduction}\label{sec_intro}


A key result in the field of ``infinite-state model-checking'' is 
Rabin's Tree Theorem \cite{rab69}. 
It says that the monadic second-order theory (short: MSO-theory) 
of the binary tree 
is decidable. Many decidability results on theories of infinite
structures have been obtained by a reduction to Rabin's Tree Theorem. 
It is also well-known that a slight extension of the signature of 
the binary tree leads to undecidability: The expansion of the binary tree by the ``equal-level relation'' $E$ 
has an undecidable monadic theory. 

The situation changes when set quantification is restricted to ``chains'', 
i.e., sets that are linearly ordered by the partial tree ordering. It is known 
(\cite{tho90}) that for the unlabeled binary tree and also for any regular binary 
tree, the chain logic theory of the tree is decidable in the presence of $E$. 
This result is of interest in verification since a large number of logical 
concepts that occur in specifications of nonterminating systems refer to 
computation paths and their subsets (i.e., to chains), for example in branching time 
logics. The second-order quantifiers in these applications do not refer to global colorings 
of computation trees (for which monadic logic would be invoked) but rather 
to quantification over chains. The equal-level relation adds the feature 
of synchronization to computation paths. 

In recent years, a theory of words and trees over infinite alphabets emerged
(\cite{nsv04,bmssd,cg09}) that opens a way for generalizations. 
Here,  a computation 
path is a sequence of letters chosen from a relational structure 
$\m = (M, R_1, \ldots, R_k)$, which is infinite in general, rather than 
from a finite alphabet $\Sigma$. Instead of the binary tree obtained from 
the words of $\{0,1\}^*$ built from the two element alphabet $\{0,1\}$, 
the infinitely branching infinite tree with vertices in $M^*$ is considered.  

There are two fundamental constructions of a tree structure built from an ``alphabet structure''
$\m$, called ``weak'', respectively ``strong'' tree iteration of 
$\m$, and denoted here ${\cal M}^\#$, respectively ${\cal M}^*$. 
For ${\cal M} = (M, R_1, \ldots, R_k)$, let
$$\Mweak = (M^*, \preceq, S, R^*_1, \ldots, R^*_k)$$
where $S(u,v)$ holds if $v = um$ for some $u \in M^*, m \in M$, 
$\preceq$ is the reflexive transitive closure of $S$, and, for $\ell$-ary $R_i$, we have  
$R^*_i(v_1, \ldots, v_\ell)$ iff for some $z \in M^*$, $v_j = z m_j$ for $j = 1, \ldots, \ell$ such
that $R_i(m_1, \ldots, m_\ell)$ holds in $\m$. 
This iteration is also called Shelah-Stupp iteration, going back 
to \cite{she75,stu75}. 

The strong tree iteration $\Mstrong$ is obtained from the weak one by adjoining 
the ``clone predicate'' 
$$C = \{u \  m \ m \mid u \in M^* , m \in M\}$$ 
to the 
signature. It allows to connect two levels of the tree structure in a way that 
``unfolding'' becomes definable. 

As shown by Shelah and Stupp \cite{she75,stu75}, respectively Muchnik and Walukiewicz (see the 
announcement in \cite{sem84} and the proof in \cite{wal02}), the MSO-theory 
of $\Mweak$ and the MSO-theory of $\Mstrong$ are decidable if the 
MSO-theory of $\m$ is. In the present paper we show the decidability of the 
chain logic theory of structures $\MweakE$, 
obtained by adjoining the equal level relation $E$ to $\Mweak$, 
under mild assumptions on the structure $\m$.   
Our results extend work of Kuske and Lohrey \cite{kl06} on structures $\Mweak$
and of B{\`e}s \cite{bes08} on structures $\MweakE$. Furthermore, we show -- in contrast 
to the Muchnik-Walukiewicz result for MSO-logic --  that a transfer 
of this decidability result to tree structures $\MstrongE$ is not possible.

B{\`e}s shows the decidability of 
the chain logic theory of $\MweakE$ if the first-order theory of $\m$ is
decidable. Here we refine his result: We refer to any logic $\el$ 
such that the $\el$-theory of $\m$ is decidable, and we consider an 
{\em extension} of the chain theory of $\Mweak$ in which further quantifications 
are allowed, namely quantifiers of $\el$ restricted to the set of siblings of 
any element $z$. (Thus one allows quantifiers over elements $y$ that are $S$-successors 
of any given element $z$.) We call the corresponding theory the {\em chain logic theory 
of $\MweakE$ with $\el$ on siblings}. We show that this theory is decidable 
if the $\el$-theory of $\m$ is. 

In our framework two logics play together: The logic $\el$ 
allows to express relations between $\m$-elements as they appear as sons 
of some given node of the tree, and chain logic is used to speak about (sets of) 
tree elements arranged along paths. Referring to the standard graphical representation
of trees, $\el$ captures the horizontal dimension and chain logic the vertical 
dimension. On the level of signatures, the predicate $E$ of the tree signature
refers to the horizontal while the successor and the prefix relation refer to the vertical aspect; 
finally,  the signature of $\m$ enters in the horizontal dimension, restricted to the 
children of a tree node. 

Standard examples of logics $\el$ are first-order logic FO, monadic second-order logic MSO 
and its weak fragment WMSO, transitive closure logic TC, or extensions of FO by counting 
operators. (In this paper we do not present a precise definition of the concept of a ``logic'' and just refer the reader to~\cite{eft07}.) Standard examples of models $\m$ originate in arithmetic and analysis, e.g.\   
$(\mathbb{N}, +, <, 0), (\mathbb{R}, +, < , 0, 1), (\mathbb{R}, +, \cdot, <, 0,1)$ (whose
first-order theory is decidable). In 
applications, one may work with structures $\m$ that are 
direct products of finite transition graphs with infinite value structures such as 
 $ (\mathbb{R}, +, < , 0, 1)$ or the real field $(\mathbb{R}, +, \cdot, <, 0,1$). 

The method to show the main result rests on a simple observation, first exploited 
in \cite{tho90}: Consider the tree with domain $M^*$ where $M$ is ordered of order type $\omega$. 
A formula $\phi(X_1, \ldots, X_n)$ of chain logic -- with chains $c_i$ 
as possible interpretations of the $X_i$ -- can be viewed as a statement about 
$2n$-tuples of $\omega$-words as follows. 
Any single chain $c_i$ is encoded by two $\omega$-words; the 
first is from $M^\omega$ and describes the (leftmost) full path of which
$c_i$ is a subset. The second is a 0-1-sequence describing by its entries 0 and 1 
which elements of the path belong to $c_i$. Now the obtained  $2n$-tuple  of $\omega$-words 
over $M$, respectively $\{0,1\}$, can be viewed as a single $\omega$-word 
with alphabet letters from $(M \times \{0,1\})^n$. Using this translation 
of $n$-tuples of chains of $\m^*$ into $\omega$-words over $(M \times \{0,1\})^n$, 
we obtain a translation of 
chain logic formulas into MSO-formulas interpreted in $\omega$-words, i.e., 
structures with domain $\mathbb{N}$. More precisely, when $\el$ is the logic 
used for $\m$, we obtain a formula of ``\ml-MSO''. 

This framework of \ml-MSO is in turn equivalent to B\"uchi automata (over 
$\omega$-words with entries from $(M \times \{0,1\})^n$). We develop these 
\ml-B\"uchi automata as a preparation for the main result. 
 It turns out that these automata allow 
closure and decidability results in precise analogy to the classical 
theory over finite alphabets. As a consequence we obtain that the chain theory of $\MweakE$ with $\el$ on siblings is decidable if the $\el$-theory of $\m$ is. 

While the setting of \ml-B\"uchi automata 
is sufficient for the study of tree models $\MweakE$, it has to be extended 
to cope with strong tree iterations $\MstrongE$ where the clone predicate
enters. We define ``strong \ml-B\"uchi automata'' for this purpose. 
Here a remarkable difference occurs between the cases of 
an input alphabet $M$ (with infinite $M$) and an input alphabet $M^n$ 
for $n>1$. We give a brief explanation that in the first case 
strong B\"uchi automata behave as \ml-B\"uchi automata (however
using just $\el =$ MSO), whereas in the second case of input alphabets $M^n$ with $n>1$, 
 undecidability phenomena enter (in the form that the emptiness 
 problem becomes undecidable). Along this line we show that the chain theory (and even the first-order theory) of $\MstrongE$ is undecidable if $\m$ is infinite -- in fact already for the case that $\m$ is the successor structure of the natural numbers.
  
A last result of the paper shows that the decidability result 
(on the chain theory of $\MweakE$ with $\el$ on siblings) 
also fails when quantification extends over an entire tree level 
rather than just siblings of a fixed node. We obtain this 
for the weak tree iteration of the two element alphabet $\{0,1\}$ 
when the logic $\el$ is MSO.   

The paper is structured as follows. In the subsequent section we 
collect the necessary terminology. Section~\ref{sec_mlaut} develops the theory 
of B\"uchi automata over $\omega$-words whose letters are 
$n$-tuples from an infinite structure $\m$ and using a logic $\el$ to specify 
properties of such letters in $\m$. 
In Section~\ref{sec_weaktree} we deduce the decidability 
of the chain theory of $\MweakE$ with $\el$ on siblings when the 
$\el$-theory of $\m$ is decidable. Section~\ref{sec_equlev} gives the two mentioned
undecidability results. We conclude with remarks on further work. 


\section{Terminology}\label{sec_term}


We consider relational structures with finite signature. Such a structure 
is presented in the format $\m = (M, R_1, \ldots, R_k)$ where $R_i$ is of arity $r_i > 0$. 
We focus on structures called ``admissible'': In this case there are 
two designated elements (usually called 0 and 1), represented by 
two singleton predicates $P_0, P_1$ that belong to the tuple $(R_1, \ldots, R_k)$.
Then we can view bit sequences as special sequences over $\m$. 

For an $\omega$-word $\alpha \in \Sigma^\omega$ (where $\Sigma$ may be infinite), 
written $\alpha = \alpha(0) \alpha(1) \ldots$,  
we denote by $\alpha[i,j]$ the segment $\alpha(i) \ldots \alpha(j)$.

We introduce two tree models built from a relational structure $\m$. 
The first is the weak tree iteration 
$$\Mweak = (M^*,\preceq, S, {R^*_1},\ldots, {R^*_k})$$ 
where 
$u \preceq v :\Leftrightarrow$ $u$ is a prefix of $v$, $S$ is the successor relation 
containing all pairs $(u, um)$ with $u \in M^*, m \in M$, and  
for every $R_i$, say of arity $\ell$, we have ${R^*_i}(v_1,\dots,v_\ell)$ iff 
there exists $z \in M^*$, $m_1, \ldots m_\ell \in M$ such that $v_j = z m_j$
for $j = 1, \ldots m_\ell$ and $R_i(m_1, \ldots, m_\ell)$. (In~\cite{bes08}
a variant of this definition is used, namely that there exist
$z_1, \ldots, z_\ell \in M^*$ of same length and $m_1, \ldots, m_\ell \in M$ 
such $v_j = z_j m_j$ with $R_i(m_1, \ldots, m_\ell)$.) 

As mentioned in the introduction, the strong tree iteration of $\m$ is the structure 
$$\Mstrong = (M^*,\preceq, S, {R^*_1},\ldots, {R^*_k}, C)$$
where everything is as above for $\Mweak$ and $C = \{u \ m \ m \mid u \in M^* , m \in M\}$. 
The expansions of $\Mweak$, $\Mstrong$ by the equal level relation $E$ (with $E(u,v)$ iff 
$|u| = |v|$) are denoted $\MweakE$, $\MstrongE$, respectively. 

If $\m$ is finite, we assume that each individual letter of $M$ is definable.
The usual approach is to introduce a constant in the signature of $\m$ for each 
element of $M$. In the present paper we stick to relational structures and 
use a singleton predicate $R_m$ for each element $m \in M$. So the binary 
alphabet $\{0,1\}$ is coded by the structure $\m_2 = (\{0,1\}, R_0, R_1)$ with 
$R_0 = \{0\}$, $R_1 = \{1\}$. In the case of finite structures $\m$ 
there is no essential difference between 
$\Mweak$ and $\Mstrong$, since the clone predicate $C$ becomes definable in $\m^\#$ by 
the equivalence 
$$C(v) \leftrightarrow \bigvee_{m \in M} (\exists u(R_m^*(u) \wedge S(u,v) \wedge R_m^*(v)).$$

Let us introduce chain logic over the tree structures $\Mweak$ and $\Mstrong$ 
built from $\m$.
A path (through the tree domain  
$M^*$) is a maximal set linearly ordered by $\preceq$; it may be identified with 
an $\omega$-word in $M^\omega$, obtained as the common extension of all 
the words $u \in M^*$ forming the path. A chain is a subset of a path. 
So a singleton set in $M^*$ is a chain, and we can easily simulate first-order quantification 
by quantification over chains restricted to singletons. We call chain logic the fragment 
of MSO logic in which set quantification is restricted to chains. 

Sometimes it is convenient to eliminate first-order variables and quantifiers in 
terms of (singleton) chain quantifiers. This simplifies the setting since only one 
kind $X_1, X_2, \ldots$ of variables remains, ranging over chains. 
In order to simulate first-order logic, the signature 
of tree models has to be adapted. As atomic formulas  one uses
\begin{itemize}
\item $\Sing(X)$ for ``$X$ is a singleton''
\item $X_i \subseteq X_j$ with its standard meaning,
\item $\Succ(X_i,X_j)$ for ``$X_i$ is a singleton $\{x_i\}$, $X_j$ is a singleton $\{x_j\}$, with $S(x_i,x_j)$;
similarly for $X_i \preceq X_j$. 
\end{itemize}
The resulting formalism is called chain$_0$ logic; it has the same expressive 
power as chain logic. 

 For an 
admissible alphabet $M$ (containing two identifiable elements 0,1) we 
encode a chain $c$ as a pair $\hat{c}:=(\alpha,\beta) \in (M^\omega)^2$ where 
  \begin{itemize}
    \item $\alpha$ encodes the path of which $c$ is a subset. As $c$ can be finite, we set $\alpha$ to be the path $m_0\ldots m_r 000\ldots$ where $m_r$ is the last $c$-element of which $c$ is a subset; it can be interpreted as a sequence of ``directions''. Note that for each element $w$ in $c$ it holds that $w$ is a prefix of $\alpha$.
    \item $\beta$ codes membership in $c$ along the path $\alpha$, i.e., $\beta(i)=1$ iff $\alpha[0,i] \in c$.
  \end{itemize}

So if $c = \emptyset$, $\alpha$ is the path $0^\omega$ through the tree $M^*$ and $\beta$ also is 
the sequence that is constant $0$. 

The technical treatment below is simplified when viewing an $n$-tuple $(\alpha_1, \ldots, \alpha_n)$ 
of $\omega$-words 
over $M$ as a single $\omega$-word over $M^n$, the \emph{convolution} of 
$(\alpha_1,\dots,\alpha_n)$:  
$$ \langle \alpha_1, \ldots, \alpha_n \rangle := \begin{bmatrix}  \alpha_1(0)\\ \vdots\\ \alpha_n(0) \end{bmatrix} 
\begin{bmatrix} \alpha_1(1)\\ \vdots\\ \alpha_n(1)\end{bmatrix} \cdots \quad \in (M^n)^{\omega}
$$
Similarly, we define the \emph{convolution of a relation $R \subseteq (M^{\omega})^n$ of 
$\omega$-words} to be the $\omega$-language 
$$L_R := \{ \langle \alpha_1, \ldots, \alpha_n \rangle  \mid 
(\alpha_1, \ldots, \alpha_n) \in R\}. $$

So the $n$-tuples of $M$-elements just considered will be used as 
letters of $\omega$-words and input letters of B\"uchi automata. Transitions 
of automata will be specified in a logic $\el$ by means of 
$\el$-formulas $\phi(x_1, \ldots, x_n)$. Each of these 
formulas defines a unary predicate $\phi^\m$ over $M^n$:  
$$\phi^\m = \{(m_1, \ldots, m_n) \in M^n \mid \m \models \phi[m_1, \ldots, m_n]\}$$

In general we consider $\omega$-models over $M^n$ for a signature that is given by a 
finite set $\Phi$ of $\el$-formulas: 
Given a tuple $(\alpha_1,\ldots,\alpha_n)$ of words over an alphabet $M$ and 
a finite set $\Phi$ of $\el$-formulas $\phi_1, \ldots , \phi_k$ with $n$ free 
variables each, we define the structure 
 $$\underline{\langle \alpha_1, \ldots, \alpha_n \rangle} = 
 (\mathbb{N}, 0, <, S, (P_{\phi})_{\phi \in \Phi})$$ 
with the usual interpretations of $0, <, S$ (the latter for the successor relation), 
and the letter predicates 
$P_{\phi_j} = \{i \in \mathbb{N} \mid (\alpha_1(i),\ldots,\alpha_n(i)) \in \phi_j^\m\}$. 
 Thus,  $P_\phi$ collects all letter positions of $\langle \alpha_1, \ldots, \alpha_n \rangle$ 
 which carry a letter from $M^n$ that shares the property described by $\phi$.

For these $\omega$-models over $\m$, equipped with predicates $P_\phi$ defined 
in $\el$, we shall use a generalized form of MSO-logic, where -- as usual in $\omega$-language theory --
the first-order quantifiers range over $\mathbb{N}$ and the monadic second-order
quantifiers over sequences of letters (here from $M$). The system will be called \ml-MSO. 

For an \ml-MSO-sentence $\psi$, where the predicates $P_\phi$ are introduced via 
$\el$-formulas $\phi(x_1, \ldots, x_n)$ with $n$ free variables, we set 
$$L(\psi) = \{\langle \alpha_1, \ldots, \alpha_n \rangle  \in (M^n)^\omega 
\mid \underline{\langle \alpha_1, \ldots, \alpha_n \rangle} \models \psi\}$$ 
as the $\omega$-language defined by $\psi$. 
We say a relation $R \subseteq (M^\omega)^n$ is \emph{\ml-MSO definable} 
if there is a \ml-MSO sentence $\psi$ with $L_R = L(\psi)$.

Later on, it will be convenient to refer to the component entries of an $\omega$-word $\langle \alpha_1, \ldots, \alpha_n \rangle$ in a more readable way than via an index $i \in \{1,\ldots,n\}$. So, when a sequence variable $Y$ is used for the $i$-th component $\alpha_i$, we shall write $Y(s)$ to indicate the element $\alpha_i(s)$ for $s \in \mathbb{N}$.

Analogous definitions can be given for the case of finite words over $M^n$. 


\section{\ml-B{\"u}chi Automata}\label{sec_mlaut}


In this section we introduce finite automata over words and $\omega$-words 
whose letters are $n$-tuples from $M$ which is the  
domain of a (in general infinite) relational structure $\m$. 
Transitions of the automata are defined in 
a logic $\el$. Mentioning both parameters (the structure $\m$ and the logic $\el$), 
we speak of \ml-automata and \ml-B\"uchi automata. In the first subsection we obtain, 
not surprisingly, an 
equivalence between \ml-automata and \ml-MSO. In the second subsection we add some 
remarks  on an extended model (``strong B\"uchi automata'') that allows to capture 
the clone predicate between successive letters. 


\subsection{The standard case}


Let $\m$ be a structure with domain $M$. 
An \emph{\ml-B{\"u}chi automaton} over $n$-tuples of $M$-elements is of the form 
$$\mathcal{B} = (Q, M^n, q_0, \Delta, F)$$ 
with a finite set $Q$ of states, the input alphabet $M^n$,  
the initial state $q_0 \in Q$, the set $F \subseteq Q$ of accepting states and 
the finite transition relation 
$\Delta \subseteq Q \times \Phi_n \times Q$, where $\Phi_n$ is the set of $\el$-formulas with $n$ free variables. 

Let us define acceptance of $\omega$-words. If 
$\alpha = \langle \alpha_1, \ldots, \alpha_n \rangle$ is an $\omega$-word over $M^n$,  
a \emph{run} of $\mathcal{B}$ on $\alpha$ is an infinite sequence of states 
$\rho = \rho(0)\rho(1) \dots$ with $\rho(0)=q_0$ such that for every $i \geq 0$ there 
exists an \ml-formula $\phi(x_1,\dots,x_n)$ and a transition $(\rho(i),\phi,\rho(i+1))$ satisfying
$$\m \models \phi[\alpha_1(i),\ldots,\alpha_n(i)]$$
A run $\rho$ of $\mathcal{B}$ on $\alpha$ is \emph{successful} 
if there exist infinitely many $i$ such that $\rho(i) \in F$. We say that $\mathcal{B}$ \emph{accepts} $\alpha$ if there exists a successful run of $\mathcal{B}$ on $\alpha$. 
We denote by $L(\mathcal{B})$ the set of $\omega$-words over $M^n$ accepted by $\mathcal{B}$.

Similarly, we define \ml-automata for the case of finite words (as done in~\cite{bes08}). 
Languages accepted by these automata will be denoted as \ml-recognizable languages. 
We note some basic properties. 
\begin{lemma}\label{lem:mlprops}\ 
\begin{itemize}
  \item The class of \ml-recognizable languages (of finite words) is closed under union, projection, and complementation.
  \item For an \ml-recognizable language (of finite words) $U \subseteq (M^n)^*$ and an \ml-B{\"u}chi recognizable $\omega$-language $K \subseteq (M^n)^{\omega}$, we have 
    \begin{enumerate}
    \item $U^{\omega}$ is \ml-B{\"u}chi recognizable.
    \item $U \cdot K$ is \ml-B{\"u}chi recognizable.
  \end{enumerate}
\end{itemize}
\end{lemma}

\begin{proof}
The closure properties of \ml-recognizable languages (of finite words) are shown by slight 
adaptions of the classical case (where the alphabet is finite). 
Here, we concentrate on pointing out the adaptions rather than the actual constructions. 
For example, an automaton for the projection from $M^n$ to $M^{n-1}$ can easily be 
obtained by replacing the ``label'' 
$\phi(x_1,\ldots,x_n)$ of a transition by $\exists x_n \phi(x_1,\ldots, x_n)$. 
For the complementation, we follow the strategy of a determinization via a powerset construction and then simply swapping the sets $F$ and $Q \setminus F$ (as 
outlined in~\cite{bes08}). The idea is as follows\eat{ (for a full work-out see~\cite{win11})}: 
Given an \ml\ automaton $\mathcal{B}$ (on finite words), $\mathcal{B}$ does not 
necessarily provide a run (accepting or not accepting) for every possible 
input letter in $M^n$, i.e., there might be a letter that does not satisfy 
any of the formulas of the transitions. For the construction of the complement automaton, 
one modifies the set of formulas for the transitions such that each input word
leads to a complete run, and additionally, one prepares for determinism: 
Let $\phi_1,\ldots,\phi_m$ be the formulas which occur in the transitions 
of $\mathcal{B}$. For each subset $J \subseteq \{1,\ldots,m\}$, introduce 
the formula $\psi_J := \bigwedge_{i \in J}{\phi_i} \wedge \bigwedge_{i \notin J}{\neg \phi_i}$. 
Note that for $J \neq K$, there is no symbol $\overline{m} \in M^n$ 
with $\m \models \psi_J \wedge \psi_K [\overline{m}]$, 
 and for each $\overline{m}$, there is a set $J$ such that $\m  \models \psi_J [\overline{m}]$. 
 Then we construct $\mathcal{B}'$ by replacing each transition 
 $(p,\phi_i,q) \in \Delta$ by $(p,\Psi_i,q)$ with $\Psi_i = \bigvee_{J \ni i} \psi_J$. 
Then $L(\mathcal{B}')=L(\mathcal{B})$, and one can continue with the usual 
powerset construction. 

Concerning the second part of the Lemma, 
for a given \ml-recognizable $U \subseteq (M^n)^*$, 
the construction of an \ml-B{\"u}chi automaton recognizing $U^{\omega}$ can be done 
in a straightforward way by isolating the initial state such that it has no incoming 
transitions and for each transition from a state $q$ to some state in $F$, adding a transition from $q$ to the initial state over the same letter, where the initial state will be the only final state in the new automaton. 
For the concatenation $U \cdot K$, we again follow a well-known idea by composing 
the two automata with additional transitions to cross over from one to the other 
at the appropriate positions.
\end{proof}

The basic decidability result on \ml-automata is the following. We state it 
for both kinds of automata: 

\begin{proposition}\label{prop:emptinessbuechi}
If the $\el$-theory of $\m$ is decidable, then the nonemptiness problem 
for \ml-automata on finite words as well as for \ml-B{\"u}chi automata is decidable.
\end{proposition}

\begin{proof}
For both kinds of \ml-automata, we have to determine whether there exists 
a word which is the label of a finite successful run. As a preparation, 
we have to check for each of the finitely many transitions $(p,\phi(x_1, \ldots ,x_n),q) \in \Delta$ 
whether it is ``useful'', i.e., whether 
there is an input letter $\overline{m} \in M^n$ satisfying 
$\phi$. This is done by invoking decidability of the $\el$-theory of $\m$, 
namely by checking whether $\m \models \exists x_1 \ldots \exists x_n \phi(x_1, \ldots, x_n)$. 
Now one considers the directed graph $(Q,R)$ where $(p,q) \in R$ if there is a useful
transition from $p$ to $q$. For an \ml-automaton over finite words, it remains to check whether 
in $(Q,R)$ there is a path from $q_0$ to $F$; for an \ml-B\"uchi automaton one verifies whether
in $(Q,R)$ there is a path from $q_0$ to a strongly connected component containing a state from $F$.
\end{proof}

We now show basic closure properties of \ml-B\"uchi automata.

\begin{lemma}\label{lem:mlbuechiclosure}
If the $\el$-theory of $\m$ is decidable, 
the class of \ml-B{\"u}chi-recognizable $\omega$-languages is effectively closed 
under union, projection, and complementation.
\end{lemma}

\begin{proof}
For union and projection the same construction as in Lemma~\ref{lem:mlprops} works. 
We sketch the construction for complementation, using the original 
approach of B\"uchi \cite{buc62}. 

Let $\mathcal{B} = (Q, M^n, q_0, \Delta, F)$ be an \ml-B\"uchi automaton.
We introduce an equivalence relation over finite $M^n$-words such that 
$(M^n)^\omega \setminus L(\mathcal{B})$ is representable as a finite union of 
sets $U \cdot V^\omega$ with \ml-recognizable sets $U,V \subseteq (M^n)^*$. 
By Lemma~\ref{lem:mlprops}, this suffices to show B\"uchi recognizability of $(M^n)^\omega \setminus L(\mathcal{B})$.

The desired equivalence relation is defined in terms of 
\emph{transition profiles}. We write for a finite word $u\in (M^n)^*$ and $p,q\in Q$: 
\begin{itemize}
 \item $\mathcal{B}: p \xrightarrow{u} q$ if there is a run on $u$ from $p$ to $q$ in $\mathcal{B}$,
 \item $\mathcal{B}: p \xrightarrow[F]{u} q$ if there is a run on $u$ from $p$ to $q$ in $\mathcal{B}$ that visits an accepting state from $F$.
\end{itemize}

A transition profile $\tau=\tp(u)$ is then given by two sets $I_{\tp(u)}$, $J_{\tp(u)}$ of pairs of states, $I_{\tp(u)}$ containing those pairs $(p,q)$ where $\mathcal{B}: p \xrightarrow{u} q$, and $J_{\tp(u)}$ containing those pairs $(p,q)$ where $\mathcal{B}: p \xrightarrow[F]{u} q$. Two words $u,v$ are called $\mathcal{B}$-equivalent, written $u \sim_\mathcal{B} v$, if $\tp(u) = \tp(v)$. This equivalence relation is of finite index: For this, note that each equivalence class (i.e., a language $U_\tau$ for a type $\tau$) is a Boolean combination of the $\ml$-recognizable languages 
$U_{pq}=\{u \mid \mathcal{B}:p \xrightarrow{u} q\}$, $U_{pq}'=\{u \mid \mathcal{B}:p \xrightarrow[F]{u} q\}$, in fact, we have
$$U_{\tau}=\bigcap_{(p,q)\in I_\tau} U_{pq} \cap \bigcap_{(p,q)\not\in I_\tau} \overline{U_{pq}} \cap \bigcap_{(p,q)\in J_\tau} U_{pq}' \cap \bigcap_{(p,q)\not\in J_\tau} \overline{U_{pq}'}.$$

Since the set of pairs $(p,q)$ is finite, we get only finitely many equivalence classes. Moreover, by Lemma~\ref{lem:mlprops} and Proposition~\ref{prop:emptinessbuechi}, we can compute those $U_\tau$ which are nonempty and hence obtain an effective presentation of the equivalence classes in terms of the corresponding finite sets $I_\tau$, $J_\tau$.

We identify the equivalence classes with the transition profiles and denote the set of these transition profiles of $\mathcal{B}$ by $\mathit{TP}_\mathcal{B}$. 
 
The following ``saturation property'' is now immediate: 

\begin{lemma}\label{saturation}
For any $\sim_\mathcal{B}$-equivalence classes $U, V$, the $\omega$-language 
$U \cdot V^\omega$ is either contained in $L(\mathcal{B})$ or in its 
complement.  
\end{lemma}

It remains to show that any $\omega$-word over $M^n$ belongs to some set 
$U \cdot V^\omega$ where $U,V$ are $\sim_\mathcal{B}$-classes. For this we 
use the transition profiles as ``colors'' of segments $\alpha[i, j]$ for $i,j \in \mathbb{N}$. 
By Ramsey's Infinity Lemma \cite{ram30} there is for any $\alpha$ and any 
B\"uchi automaton $\mathcal{B}$ a pair of transition profiles $\tau_0, \tau$ from 
$\mathit{TP}_\mathcal{B}$ and an infinite set $I = \{i_0 < i_1 < i_2 < \ldots \}$ such that 
$$ \tp(\alpha[0,i_0 -1]) = \tau_0, \ \ \tp(\alpha[i_j, i_{j+1}-1]) = \tau \text{ for }j \geq 0.$$ 
This shows that $\alpha \in U_{\tau_0} \cdot U^\omega_{\tau}$, where $U_{\tau_0}$, $U_{\tau}$ denote the equivalence classes of $\sim_{\mathcal{B}}$ corresponding to $\tau_0$ resp.\ $\tau$. 
Let  
$$\mathit{NTP}_\mathcal{B} = \{(\tau_0,\tau) \in \mathit{TP}_\mathcal{B}^2 \mid 
U_{\tau_0} \cdot U_{\tau}^{\omega} \cap L(\mathcal{B}) = \emptyset\}$$ 
Again, by decidability of the $\el$-theory of $\m$, this set is computable. 
Then 
$$(M^n)^{\omega} \setminus  L(\mathcal{B}) = \bigcup_{(\tau_0,\tau) \in \mathit{NTP}}{U_{\tau_0}U_{\tau}^{\omega}}.$$
\ 
\end{proof}

As a consequence of Lemma~\ref{lem:mlprops} and Lemma~\ref{lem:mlbuechiclosure} we obtain the following result. 

\begin{proposition}\label{prop:inclusionEquivBuechi}
If the $\el$-theory of $\m$ is decidable, the inclusion problem and the equivalence problem for 
\ml-B{\"u}chi recognizable languages are decidable. 
\end{proposition}

After these preparations, one can easily infer an equivalence between \ml-B{\"u}chi automata 
and \ml-MSO.

\begin{remark}\label{rem:AtoFBuechi}
Let $\mathcal{B}$ be an \ml-B{\"u}chi automaton, then there exists an \ml-MSO sentence $\psi$ with $L(\mathcal{B}) = L(\psi)$.
\end{remark}

Again, the construction of an \ml-MSO formula describing a successful run of a given \ml-B{\"u}chi automaton $\mathcal{B}$ is a straightforward adaption of the well-known proof (\cite{tho97}). 
The only modification occurs in the formulas describing the transitions of $\mathcal{B}$: 
for  a transition $(p, \phi, q)$, one uses the predicates $P_{\phi}(x)$ as introduced above 
in the definition of \ml-MSO.

Let us turn to the translation from \ml-MSO sentences to \ml-B{\"u}chi automata. 

\begin{proposition}\label{prop:FtoABuechi}
Let $\psi$ be an \ml-MSO sentence, then there exists an \ml-B{\"u}chi-automaton $\mathcal{B}$ with $L(\psi) = L(\mathcal{B})$.
\end{proposition}

\begin{proof}
We first modify \ml-MSO to the expressively equivalent formalism of \ml-MSO$_0$-formulas in complete 
analogy to the definition of chain$_0$ logic in Section~\ref{sec_term}. We proceed by induction over MSO$_0$-formulas. 

For the induction basis, we consider the atomic formulas $X_i \subseteq X_j$, $\Sing(X_i)$, $\Succ(X_i, X_j)$, $X_i \preceq X_j$, and $X_i \subseteq P_{\phi}$ and specify \ml-B{\"u}chi automata that recognize the sets of $\omega$-words defined by these formulas. To exemplify, we give the automaton for $X_i \subseteq P_{\phi}$, which checks that when the $i$-th component is $1$, the letter vector satisfies the \ml-formula $\phi$, which defines the letter predicate $P_{\phi}$.
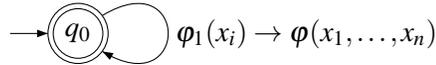
\begin{center}
  \setlength{\unitlength}{.8ex}
  \begin{picture}(6,8)
    \label{fig:XsubsetP}
    \gasset{Nframe=y,AHnb=1}
    \node[Nmarks=ir](q0)(3,4){$q_0$}
    \drawloop[loopangle=0](q0){$\phi_1(x_i) \to \phi(x_1,\ldots,x_n)$}
  \end{picture}
\end{center}
For the induction step, we consider the connectives $\vee$ and $\neg$, as well as the existential quantifier $\exists$. Here, we can exploit the closure properties of \ml-B{\"u}chi automata from Lemma~\ref{lem:mlbuechiclosure}, and employ the constructions for the union, complementation, and projection, respectively.
\end{proof}

As a relation $R \subseteq (M^{\omega})^n$ is representable by a convolution as an $\omega$-word over $M^n$, Remark~\ref{rem:AtoFBuechi} and Proposition~\ref{prop:FtoABuechi} yield the following result.

\begin{theorem}\label{theo:maintheobuechi}
A relation $R \subseteq (M^{\omega})^n$ with $n\geq 1$ of $\omega$-words is \ml-MSO definable  iff 
it is \ml-B{\"u}chi-recognizable. The transformation in both directions is effective.
\end{theorem}

As a consequence of the \ml-B{\"u}chi theory, we obtain that satisfiability and 
equivalence of \ml-MSO-formulas over models from $M^\omega$ are 
decidable if the $\el$-theory of the structure $\m$ is decidable.


\subsection{Strong $\m$-$\el$-B\"uchi automata}


In the second part of this section, we extend -- as far as possible -- 
the techniques and results to a 
slightly stronger model of B\"uchi automaton. While the B\"uchi automata above 
are appropriate for treating the structures $\MweakE$, a stronger model is motivated by 
the study of strong tree iterations $\MstrongE$ in which the clone predicate 
enters. Recall that it allows to single out those elements of $M^*$ which are 
of the form $u \ m \ m$. Thus, when reading a ``letter'' $m$ along a path, we need to incorporate 
the feature to ``remember'' whether this current input letter $m$ coincides with the 
previous one.  

We define the notion of {\em strong $\m$-$\el$-B\"uchi automaton} over 
$n$-tuple input letters (i.e., with input alphabet $M^n$, $M$ being the domain of $\m$). The format is the same as for standard B{\"u}chi automata over $M^n$ as mentioned above, except for the transitions. 
For each state pair $(p,q)$ 
 the possible transitions are defined by a formula $\phi_{pq}(x_1, \ldots, x_n, 
y_1, \ldots , y_n)$ -- or, in the special case of an initial transition, by 
a formula $\phi_{q_0 q}(x_1, \ldots x_n)$. Starting with the latter case, 
the automaton can proceed from $q_0$ to $q$ with input letter $(m_1, \ldots, m_n)$ 
if $\m \models \phi[m_1, \ldots, m_n]$. For a transition of the first case, 
in which a previous input letter exists and is $(m^-_1, \ldots, m^-_n)$, 
the automaton can move from $p$ to $q$ if $\m \models \phi[m_1, \ldots, m_n, m^-_1, \ldots, m^-_n]$.
All other notions are copied from the case of (standard) $\m$-$\el$-B\"uchi automata. 

We can reprove the basic decidability and closure properties only under rather radical 
restrictions, namely just for the logic $\el$ = MSO 
and for the case of input letters from $M$ (rather than $n$-tuples of such 
letters). We only give a rough outline; in the present paper 
we do not apply these automata to chain logic over tree structures. 
  
First let us state the basic decidability result.

\begin{lemma}
If the MSO-theory of $\m$ is decidable, 
the emptiness problem for strong  $\m$-MSO-B\"uchi automata over $M$ 
is decidable. 
\end{lemma}

\begin{proof}
The proof of this lemma can either be given directly, or by invoking the 
above-mentioned Muchnik-Walukiewicz result (\cite{sem84,wal02}). 
It states -- under the assumption that 
the MSO-theory of $\m$ is decidable -- that the MSO-theory of $\Mstrong$ is 
decidable. The nonemptiness of a strong B\"uchi automaton over $M$ can be 
decided by checking existence of a suitable path through $\Mstrong$. 
\end{proof}

\begin{lemma}
If the MSO-theory of $\m$ is decidable, the 
class of $\omega$-languages recognized by strong $\m$-MSO-B\"uchi automata 
over $M$ is effectively closed under the Boolean operations and 
definable projections $p: M \rightarrow M$.
\end{lemma}

\begin{proof}
This claim is shown in precise analogy to the case of standard 
$\m$-$\el$-B\"uchi automata (and we skip here the repetition of proofs), 
except for the closure under complement. Here we describe the necessary 
modifications. 

The approach is the same as for the standard case, 
i.e., via B\"uchi's original method involving finite colorings and 
Ramsey's Theorem. However, the coloring of a segment of an 
$\omega$-word over the alphabet $M^n$, i.e., the transition profile, is 
defined differently. Given a strong B\"uchi automaton ${\cal A}$, 
the ``strong transition profile'' 
of the segment $\alpha[i,j]$ 
of an $\omega$-word $\alpha$ refers also to the 
last previous letter $\alpha(i-1)$ if $i > 0$. This extra context 
information is needed in order to capture the clone predicate on the 
$n$ components of $\alpha$, and we define the transition profile of a segment 
relative to this context information within $\alpha$. So an appropriate 
notation for a strong transition profile is $\tp_\alpha([i,j])$ rather than 
$\tp(u)$. Such profiles, however, are of the same type as the previously 
defined profiles (namely, presented as two sets of pairs of states). The
transition profile of a segment $\alpha[i,j]$ is fixed from the 
state pairs $(p,q)$ that allow a run of the automaton from $p$ to $q$ 
(respectively, a run from $p$ to $q$ via a final state), where   
in the first move the letter $\alpha(i-1)$ is used. (This condition 
is dropped for the case $i=0$.) 

There is, of course, a definite conceptual difference to the usual 
coloring of segments in terms of standard transition profiles: 
There, one may concatenate any sequence of segments (for given transition 
profiles) to obtain a new composed segment whose transition profile 
is induced by the given ones. In the new setting, the composition 
of segments $u$ and $v$ only works when the clone information on 
the last letter of $u$ agrees with the first letter of $v$. However, 
this does not affect the argument in B\"uchi's complementation proof: 
Here we only need that for any {\em given} $\alpha$ one can obtain 
a sequence $i_0 <  i_1 <  \ldots$ such that all segments 
$\alpha[i_j,i_{j+1} -1]$ share the same transition profile, 
and that for such a sequence, the transition profiles of 
$\alpha[0, i_0 -1]$ and of $\alpha[i_0, i_1 - 1]$ determine 
$\alpha$ either to be accepted of not to be accepted by 
the B\"uchi automaton.  

Also the sets $U_{\tau_0} \cdot U_\tau^\omega$ can be used as before
when defined properly:  
Such a set is {not} obtained by freely concatenating 
a segment $u \in U_{\tau_0}$ and a sequence of segments from $U_\tau$; 
rather, it is the set 
$$U_{\tau_0} \cdot U_\tau^\omega = \{\alpha \mid \exists i_0, i_1, \ldots (0 < i_0 < i_1 < \ldots \wedge 
\tp_\alpha[0,i_0-1] = \tau_0 \wedge  \tp_\alpha[i_j, i_{j+1}-1] = \tau \text{ for }j= 0,1, \ldots)\}$$
The effective presentation of the complement of $L({\cal A})$ is now completed
as in the preceding subsection for \ml-B\"uchi automata.  
\end{proof}

In Section~\ref{sec_equlev} below we shall see that these results fail for the case of an infinite alphabet $M^n$ with infinite $M$ and $n>1$.


\section{Weak Tree Iterations}\label{sec_weaktree}


In this section, we want to show that for the weak tree iteration with 
equal level relation, the chain theory with $\el$ on siblings is decidable if the 
$\el$-theory of $\m$ is. 

With the preparations of Section~\ref{sec_mlaut}, we will establish 
a reduction from chain logic formulas over tree models to \ml-MSO 
over $\omega$-sequences (and then to B\"uchi automata).

To avoid heavy notation, we employ chain$_0$ logic as introduced in 
Section~\ref{sec_term}, and provide the following construction. Recall that for 
a chain $c$ in $\MweakE$, the object $\hat{c}$ is a pair of sequences over $M$
coding the path underlying the chain $c$, respectively the membership 
of nodes of this path in $c$. 
 
\begin{lemma}
For any chain$_0$-formula  $\phi(X_1,\dots,X_n)$ over $\MweakE = (M^*,S,\preceq,R_1^*,\dots,R_k^*,E)$ 
with $\el$ on siblings, one can construct an $\ml$-MSO-formula $\phi'(Y_1,Z_1,\ldots,Y_n,Z_n)$ interpreted in $\omega$-words over $M^{2n}$ such that for all chains $c_1,\dots,c_n$ we have:
\begin{center}
 $\MweakE \models \phi [c_1,\dots,c_n]$\\
 $\textrm{if and only if }\underline{\langle \hat{c}_1,\ldots,\hat{c}_n\rangle} \models \phi'(Y_1,Z_1,\ldots,Y_n,Z_n).$
\end{center}
\end{lemma}

\begin{proof}
We proceed by induction over the structure of chain$_0$-formulas with $\el$ on siblings over $\MweakE$.

For the induction basis we have to consider the atomic formulas, namely of the form 
$\Sing(X)$, $X_i \subseteq X_j$, $X_i \preceq X_j$, $R_i^*(X_1,\dots,X_k)$, $E(X_i,X_j)$, and also the $\el$-formulas $\gamma(x_{i_1},\dots,x_{i_\ell})$.

As a first example, we present the translation into $\ml$-MSO-formulas for the formula $\phi(X)=\Sing(X)$: Given the encoding $\hat{c}=(\alpha,\beta)$ of a chain $c$, the formula $\phi_{\Sing}'(X)$ has to express that $\beta$ indicates membership in $c$ exactly once. Thus, we obtain $\phi_{\Sing}'(Y,Z)=\exists s\big(Z(s) \wedge \forall t(t \neq s \to \neg Z(s))\big)$. 

For the case of an $\el$-formula $\gamma(x_{i_1},\dots,x_{i_\ell})$, we capture $x_{i_1},\ldots,x_{i_\ell}$ by corresponding singletons $X_{i_1},\ldots,X_{i_\ell}$, and these in turn by pairs $(Y_{i_1},Z_{i_1}),\ldots,(Y_{i_\ell},Z_{i_\ell})$ consisting of a path $Y_{i_j} \in M^\omega$ and a singleton set indicator $Z_{i_j} \subseteq \{0,1\}^\omega$ each. We have to define a corresponding predicate $P_\gamma \subseteq ((M \times \{0,1\})^n)^\omega$ by an $\ml$-MSO-formula\eat{$\phi_{P_\gamma}'$} that expresses in terms of the $Y_{i_j}$, $Z_{i_j}$ that there is a common $S$-predecessor $z$ of the elements $x_{i_j}$ and that the tuple $x_{i_1},\ldots,x_{i_\ell}$ satisfies $\gamma$. In intuitive notation, we have
$$
\phi_{P_\gamma}'(Y_1,Z_1,\ldots,Y_n,Z_n)= \bigwedge_{j=1}^\ell \text{``$(Y_{i_j},Z_{i_j})$ is singleton containing $x_{i_j}$''} \wedge \exists z \bigwedge_{j=1}^\ell \text{``}S(z,x_{i_j})\text{''} \wedge \gamma(x_{i_1},\ldots,x_{i_\ell})
$$

In some more detail:
\begin{multline*}
\bigwedge_{j=1}^\ell \phi_{\Sing}'(Y_{i_j},Z_{i_j}) \wedge \exists x_{i_1}\ldots\exists x_{i_\ell}\exists s \big(Z_{i_j}(s) \wedge Y_{i_j}(s)=x_{i_j} \wedge \bigwedge_{j' \neq j} \forall t < s\, (Y_{i_j}(t)=Y_{i_{j'}}(t)) \wedge \gamma(x_{i_1},\ldots,x_{i_\ell}) \big)
\end{multline*}

The induction step then is straightforward, as $\ml$-MSO is closed under the Boolean operations and projection. 
\end{proof}
Thus, we obtain a reduction of the chain$_0$-theory  with $\el$ on siblings of $\MweakE$ to the $\ml$-MSO theory, which with Theorem~\ref{theo:maintheobuechi} is decidable if the $\el$-theory of $\m$ is decidable. This leaves us to conclude this section with the following theorem:

\begin{theorem}\label{theo:weaktreeiterations}
If the $\el$-theory of $\m$ is decidable, the 
chain-theory of $\MweakE$ with $\el$ on siblings is decidable.
\end{theorem}


\section{Undecidability Results}\label{sec_equlev}


In the previous sections we showed decidability of the model-checking problem
for chain logic with $\el$ on siblings over tree structures $\MweakE$, given a 
structure $\m$ with decidable $\el$-theory for some logic $\el$. 

The first result of this section shows that this does not 
extend to strong tree iterations $\MstrongE$ (even if we confine ourselves 
to first-order logic in place of chain logic). 

The second result shows another limitation to decidability:
In the  ``horizontal dimension'' of tree models, we may (in Theorem \ref{theo:weaktreeiterations})
use 
 $\el$-quantifiers ranging over children of given nodes. We show 
  that for the case $\el =$ MSO we lose decidability when the horizontal quantification 
 is extended to an entire tree level. Here we get undecidability for the weak 
 tree iteration.
 
For the first result we use a reduction 
from the termination problem of 2-counter machines (or 2-register machines). 
Such a machine $M$ is given by a finite sequence 
$$ 1 \ {\rm instr}_1; \ldots; k-1 \ {\rm instr}_{k-1}; k \  {\rm stop}$$
where each instruction ${\rm instr}_j$ is of the form 
\begin{itemize}
\item 
Inc$(X_1)$, Inc$(X_2)$ 
(increment the value of $X_1$, respectively $X_2$ by 1),  or
\item 
Dec$(X_1)$, Dec$(X_2)$ (similarly for decrement by 1, with the 
convention that a decrement of 0 is 0), or 
\item 
If $X_i = 0$ goto $\ell_1$ else to $\ell_2$ (where $i = 1,2$ and $ 1 \leq 
\ell_1, \ell_2 \leq k$, with the natural
interpretation).
\end{itemize}
An $M$-configuration is a triple $(\ell, m,n)$, indicating 
that the $\ell$-th instruction is to be executed and the values 
of $X_1, X_2$ are $m,n$, respectively. A terminating 
$M$-computation (for $M$ as above) 
is a sequence $(\ell_0, m_0, n_0), \ldots, 
(\ell_r, m_r, n_r)$ of $M$-configurations where in each step the update is done
according to the instructions in $M$ and the last 
instruction is the stop-instruction (formally: $\ell_r = k$). 
The termination problem for 2-counter machines asks to 
decide, for any given 2-counter machine $M$, whether 
there exists a terminating $M$-computation that starts with 
$(1,0,0)$ (abbreviated as $M : (1,0,0) \rightarrow {\rm stop}$).
It is well-known that the termination problem for
2-counter machines is undecidable (\cite{min67}).

We turn to the model-checking problem over structures $\MstrongE$. We show 
undecidability when 
$\m$ is the structure ${\cal S} := (\mathbb{N}, \Suc)$ (where $\Suc$ is successor). 

\begin{theorem}
The first-order theory of ${\cal S}_E^*$ with FO on siblings is undecidable. 
\end{theorem}

\begin{proof}
For any 2-register machine $M$ we construct a first-order formula 
$\phi_M$ with FO on siblings such that $M : (1,0,0) \rightarrow {\rm stop}$ 
iff ${\cal S}_E^* \models \phi_M$. 

The idea is to code a computation $(\ell_0, m_0, n_0), \ldots, 
(\ell_r, m_r, n_r)$
by three finite paths of same length, one for each of the 
three components. Each of these paths (namely $\pi_0 = (\ell_0, \ldots, \ell_r), 
\pi_1 = (m_0, \ldots, m_r), \pi_2 = (n_0, \ldots, n_r)$) is determined by its 
last point in the tree structure $\mathcal{S}_E^*$, i.e., by a triple $x_0$, $x_1$, $x_2$ of $\mathcal{S}_E^*$-elements.

We use a formula which expresses
$$\exists x_0 \exists x_1 \exists x_2 (E(x_0, x_1) \wedge E(x_1, x_2) \wedge [x_0, x_1, x_2
\mbox{ code a terminating computation of } M]).$$  

In order to obtain a formalization of the condition in squared brackets, we have to express 
\begin{enumerate}
\item 
the initial condition that $\pi_0$ starts with the son $1$ of the root 
and $\pi_1, \pi_2$ with the son $0$ of the root, 
\item
the progress condition that for each $y_0 \prec x_0$ (giving an instruction number), the corresponding 
$M$-instruction is executed, which involves the vertex $y_0$ and the vertices $y_1 \prec x_1,
y_2 \prec x_2$ on the same level as $y_0$ and their respective successors $z_0, z_1, z_2$ 
on $\pi_0, \pi_1, \pi_2$, respectively, 
\item 
the termination condition that $x_0$ is the number $k$.
\end{enumerate}

Accordingly, we can formalize the condition in squared brackets by a conjunction of three formulas $\phi_1$, $\phi_2$, $\phi_3$ in the free variables $x_0$, $x_1$, $x_2$, making use of the (definable) tree successor relation $S$. 
\begin{itemize}
  \item The formula $\phi_1$ expresses (in first-order logic with FO on siblings) for the root $r$ of the tree model and those three $S$-successors $y_0$, $y_1$, $y_2$, where $y_0 \preceq x_0$, $y_1 \preceq x_1$, $y_2 \preceq x_2$, that $y_0$ is the number $1$ and $y_1$, $y_2$ are the number $0$ (of the model ${\cal S}=(\mathbb{N},\Suc)$).
  \item The formula $\phi_2$ is of the form: 
  \begin{quotation}\noindent
    ``for all $y_0 \prec x_0$, $y_1 \prec x_1$, $y_2 \prec x_2$ with $E(y_0,y_1)$ and $E(y_0,y_2)$, there are tree-successors $z_0$, $z_1$, $z_2$ (i.e., with $S(y_0,z_0)$, $S(y_1,z_1)$, $S(y_2,z_2)$ with $z_0 \preceq x_0$, $z_1 \preceq x_1$, $z_2 \preceq x_2$) that represent the correct update of the configuration $(y_0,y_1,y_2)$.''
  \end{quotation}
  The condition on update is expressed by a disjunction over all program instructions; we present, as an example, the disjunction member for the statement ``$3$~Inc$(X_2)$'':
  \begin{multline*}
    y_0 \text{ is number $3$ in } (\mathbb{N},\Suc) \to z_0 \text{ is number $4$ in } (\mathbb{N},\Suc)\\
    \wedge z_1 \text{ is the clone of }y_1 \wedge z_2 \text{ is the $\Suc$-successor of the clone of }y_2.
  \end{multline*}
  It is easy to formalize this in first-order logic with FO on siblings, similarly for the Dec-instructions and the jump instructions.
  \item The formula $\phi_3$ expresses the third condition and is clearly formalizable in first-order logic with FO on siblings.
\end{itemize}
\end{proof}

This result can also be stated in the framework of strong 
B\"uchi automata (or even strong automata on finite words) when 
the alphabet consists of pairs of natural numbers: 
With each 2-register machine $M$ one associates 
a strong $\mathcal{S}$-MSO-automaton $\mathcal{A}_M$ over $\mathbb{N}^2$
which accepts an input word $(m_0,n_0) \ldots (m_r,n_r)$ if this represents 
the sequence of register values of a terminating computation of $M$; the 
existence of an appropriate sequence of instruction numbers (from $\{1, \ldots, k\}$)
can be expressed by a block $\exists X_1 \ldots \exists X_k$ of MSO-quantifiers. 
(In fact, weak MSO-quantifiers suffice.)    
   
Let us turn to the second undecidability result. We shall confine ourselves to the 
 simplest setting, where the structure $\m$ is just $(\{0,1\}, \{0\}, \{1\})$, 
 i.e., $\MweakE$ and $\MstrongE$ are \emph{both} the binary tree with 
 equal level relation (see also~\cite{tho09}). 
 
\begin{theorem}
The chain theory of the binary tree with equal level relation and 
MSO on tree levels is undecidable. 
\end{theorem}

\begin{proof}
We use an idea of~\cite{pt93} that allows to 
code a tuple of finite sets of the binary tree up to (and excluding) level $L$ 
by a tuple of subsets of level $L$ itself. In other words, we code 
a subset $S$ of tree nodes before level $L$ by an ``antichain'' $A$ which is a subset of the level $L$ (see Figure~\ref{fig:coding}). 

\begin{figure}
\begin{center}
  \setlength{\unitlength}{.8ex}
  \begin{picture}(40,24)(6,10)
    \gasset{Nframe=y,AHnb=0,Nadjust=wh}
    \drawpolygon(6,12)(21,34)(36,12)			
    \drawccurve(16,19)(20,26)(24,22.5)(22.5,20)		
    \drawccurve(8,12)(16,13.5)(24,12)(16,10.5)		
    \gasset{Nframe=n}
    \node(bul)(22,22.5){$\bullet$}
    \node(v)(20,22.5){$v$}
    \node(vp)(17.5,10){$v'$}
    \drawline[AHnb=1](23,22)(25,20.5)(17.5,12)		
    \node(dummy1)(20,24){}
    \node(S)(41,28){$S$}
    \drawedge[curvedepth=3](dummy1,S){}
    \node(dummy2)(22,12){}
    \node(A)(41,17){$A$}
    \drawedge[curvedepth=3](dummy2,A){}
    \put(40,33){level $0$}
    \put(40,12){level $L$}
  \end{picture}
  \caption{\label{fig:coding}Coding an element of a set $S$ by an element of an antichain $A$.}
\end{center}
\end{figure}
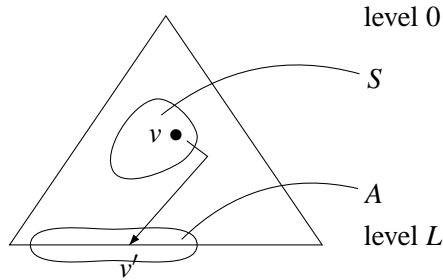 

We simply map a vertex $v$ (before level $L$) to the unique vertex $v' \in L$ which 
belongs to $v 1 0^*$ (i.e., belongs to the leftmost path from 
the right successor of $v$; see again Figure~\ref{fig:coding}). The map $v \mapsto v'$ is injective
and definable in chain logic (even in FO-logic), given the level $L$. 
Moreover, it is easy to see that the relations of being 
left or right successor in the tree are translated to FO-definable 
relations over the level $L$ under consideration. 

Using this coding, an existential quantifier over finite sets  
in the binary tree is captured by an existential quantifier 
over subsets of an appropriate level of the tree (namely, of  
a level that is beyond all maximal elements of the finite set under
consideration). 

Thus, the weak MSO-theory of the binary tree with $E$ is interpretable 
in the FO-theory of the  binary tree $(\{0,1\}^*, S_0, S_1, \preceq, E)$ 
with $E$ and with MSO restricted to levels.
 
Since the weak MSO-theory of the binary tree with $E$ is undecidable 
(see e.g.~\cite{tho90b}), we obtain the claim.  
\end{proof}


\section{Conclusion}\label{sec_con}


In this work, we outlined a theory of generalized B\"uchi automata 
over infinite alphabets. These alphabets are represented by relational structures 
$\m$, the transitions being specified by formulas of a logic $\el$ over $\m$. 
In this setting of $\m$-$\el$-B\"uchi automata (which only slightly generalizes 
that of \cite{bes08}), the nonemptiness problem becomes decidable if the $\el$-theory 
of $\m$ is. An extended model of strong $\m$-$\el$-B\"uchi automata was introduced 
in which a transition via an $\m$-input may depend on the previous $\m$ input. Here 
an essential difference appears between the cases where input letters are from $M$ 
and where  input letters are in $M^n$ for $n> 1$. 

We applied this theory to show that the chain logic theory of the weak tree 
iteration $\MweakE$ of $\m$ (with $\el$ chosen as above) 
is decidable where the equal level relation is adjoined, and quantifications 
of $\el$ over siblings of the tree model are allowed. On the other hand, we showed limits 
for generalization. For example, we showed undecidability 
for the corresponding theory of the strong tree iteration when the underlying 
model is the successor structure of the natural numbers.  

Several problems are raised by this study. Since the logics considered here all 
have nonelementary complexity, it may be interesting to set up fragments and ``dialects''
(e.g.\ in temporal logics) of chain logic where the complexity is better. 
Also, it seems that variants of the model of strong (B\"uchi-) automaton 
should be studied in more depth, for instance by an integration with the 
theory of automata over ``data words'' as developed in \cite{nsv04,bmssd,cg09}.


%


\bibliographystyle{eptcs}
\bibliography{infbib}
\end{document}